\documentclass[]{article}

\usepackage[USenglish]{babel} 
\usepackage[T1]{fontenc}
\usepackage[ansinew]{inputenc}

\usepackage{lmodern} 

\usepackage[a4paper]{geometry}

\usepackage{listings}             
\usepackage{graphicx} 
\usepackage{caption}
\usepackage{subcaption}
\usepackage{cite}
\usepackage{tikz}	
\usetikzlibrary{arrows,automata,shadows}
\usetikzlibrary{matrix}
\usetikzlibrary{chains}
\usetikzlibrary{positioning}
    \everymath{\displaystyle}
    \tikzstyle{mathbox} = [inner sep=0pt, anchor=base]
    \tikzstyle{every picture}+=[remember picture]
    \usetikzlibrary{calc}
    \usetikzlibrary{positioning}
    \xdefinecolor{darkgreen}{RGB}{175, 193, 36}
    \xdefinecolor{darkgreen}{RGB}{175, 193, 36}
 \xdefinecolor{mdhOrange}{RGB}{255, 136, 0}
 \xdefinecolor{mdhTurkos}{RGB}{0, 153, 153}

\newenvironment{remark}[1][Remark]{\begin{trivlist}
\item[\hskip \labelsep {\bfseries #1}]}{\end{trivlist}}

\usepackage{hyperref}
\hypersetup{
colorlinks=false,
linkbordercolor={1 1 1}, 
citebordercolor={1 1 1}, 
urlbordercolor={1 1 1} 	 
}

\usepackage{amsmath}
\usepackage{amsthm}
\usepackage{amsfonts}

\usepackage{amssymb} 
\newtheorem{theorem}{Theorem}[section]

\newtheorem{corollary}{Corollary}[section]
\newtheorem{definition}{Definition}[section]

\title{Graph partitioning and a componentwise PageRank algorithm}

\author{
 Christopher Engstr{\"o}m,\\
Division of Applied Mathematics\\
Education, Culture and Communication (UKK), M{\"a}lardalen University,\\
christopher.engstrom@mdh.se\\[0.5cm] 
Sergei Silvestrov\\
Division of Applied Mathematics\\
Education, Culture and Communication (UKK), M{\"a}lardalen University\\
sergei.silvestrov@mdh.se}

\begin{document}
\maketitle

\date{}

\abstract
In this article we will present a graph partitioning algorithm which partitions a graph into two different types of components: the well-known `strongly connected components' as well as another type of components we call `connected acyclic component'. We will give an algorithm based on Tarjan's algorithm for finding strongly connected components used to find such a partitioning. We will also show that the partitioning given by the algorithm is unique and that the underlying graph can be represented as a directed acyclic graph (similar to a pure strongly connected component partitioning). 

In the second part we will show how such an partitioning of a graph can be used to calculate PageRank of a graph effectively by calculating PageRank for different components on the same `level' in parallel as well as allowing for the use of different types of PageRank algorithms for different types of components. 

To evaluate the method we have calculated PageRank on four large example graphs and compared it with a basic approach, as well as our algorithm in a serial as well as parallel implementation.

\tableofcontents

\pagestyle{plain} 

\newpage

\section{Introduction}
The PageRank algorithm was initially used by S. Brinn and L. Page to rank homepages on the Internet \cite{tAnatomy-largeWebSearch}. While the original algorithm itself is already very efficient, given the sheer size and rate of growth of many real life networks there is a need for even faster methods. Much is known of the original method using a power iteration of a modified adjacency matrix such as how the damping factor $c$ affect the condition number and convergence speed \cite{ilprints582,ilprints597}. 

Many ways have been proposed in order to improve the method such as aggregating webpages that are in some way `close' \cite{5399514} or by excluding webpages whose rank are found to already have converged from the iteration procedure \cite{Kamvar200451}. Another method to speed up the algorithm is to remove so called dangling pages (pages with no links to any other page), and then calculate their rank at the end separately \cite{FAndersson_art_PR,journals/im/LeeGZ07}. A similar method can also be used for root vertices (pages with no links from any other pages) \cite{Qing_lumping2}. 
A good overview over different methods for calculating PageRank can be found in \cite{Berkhin05asurvey}. 

The method we propose here have some similarities with the method proposed by \cite{Qing_lumping2}, the main difference being that we work on the level of graph components rather than single vertices. Another method with a similar idea is the one proposed by \cite{Arasu_et_at}, but we use two types of components as well as look at using different PageRank algorithms for different types of components and how different components can be calculated in parallell.

We presented some parts of this work at ASMDA 2015 \cite{ce_ASMDA2015} while in this paper we improve the method further as well as implementing and presenting some results using the method. 

The rest of this paper is organized as follows: In Sec. \ref{sec:graph_concepts} we define some graph concepts as well as define and prove a couple of properties of the graph partitioning which will be used later in our PageRank algorithm. In Sec. \ref{sec:PageRank} we define the variation of PageRank that will be used as well as show how PageRank can be computed for different types of vertices or graph components. In Sec. \ref{sec:method} we describe our proposed algorithm to calculate PageRank, first by describing how to efficiently find the graph partitioning described earlier and next how to use this to calculate PageRank. In the same section we also verify the linear computational complexity of the algorithm as well as give some error estimates in Sec. \ref{sec:error}. At the end of our paper in Sec. \ref{sec:experiments} we describe our implementation of the method as well as show some results using the method on a couple of different graphs.

\section{Notation and Abbreviations}
The following abbreviations will be used throughout the article
\begin{itemize}
\item SCC - strongly connected component. (see Def. \ref{def:SCC})
\item CAC - connected acyclic component. (see Def. \ref{def:CAC})
\item DAG - directed acyclic graph.
\item BFS - breadth first search.
\item DFS - depth first search.
\end{itemize}
Explanation of notation. 
\begin{itemize}
\item We denote by $G$ a graph with with vertex set $V$ and edge set $E$, by $|V|$ the number of vertices and $|E|$ the number of edges. 
\item Vectors are denoted by an arrow for example $\vec{V},~ \vec{u}$ and matrices in bold capital letters such as ${\bf A},~{\bf M}$.
\item Each vertex in a graph has an assigned level $L$ (see Def. \ref{def:Level}), $L^+$ denotes all vertices of level $L$ or greater while $L^-$ denotes all vertices of level $L$ or less. For example if ${\bf A}$ is the adjacency matrix of some graph, then ${\bf A}_{L^+,(L-1)^-}$ denotes the submatrix of ${\bf A}$ corresponding to all rows representing vertices of level $L$ or greater and all columns representing vertices of level $L-1$ or less. 
\item Different versions of PageRank will be defined where the type is denoted inside parenthesis, for example $\vec{R}^{(1)}$ (see Def. \ref{def:R1} and Def. \ref{def:R3}). 
\item Lowercase letters as subscripts denote single elements while capital letters denote a set of elements for example $w_i$ denoting the element of $\vec{W}$ corresponding to vertex $v_i$ and $\vec{W}_L$ denoting the vector of elements corresponding to all vertices of level $L$. 
\item $P(v_i\rightarrow v_j)$ denotes the probability to hit vertex $v_j$ in a random walk starting at $v_i$ (see Def. \ref{def:R3prob}). 
\end{itemize}
\section{Graph concepts}\label{sec:graph_concepts}
In this work we assume all graphs to be simple directed graphs. Since only the adjacency matrix of any graph will be used, all graphs are also assumed to be unweighted. It should be noted however that it would be fairly simple to generalize for certain types of weighted graphs (such as those representing a Markov chain).
We will start by defining two types of graph components, one being the well-known `strongly connected component' and another we call `connected acyclic component'. We also define the notion of `level' of a component representing the depth of the component in the underlying graph where components are represented by a single vertices.  
\begin{definition}\label{def:SCC}
A strongly connected component (SCC) of a directed graph $G$ is a subgraph $S$ of $G$ such that for every pair of vertices $u,v$ in $S$ there
is a directed path from $u$ to $v$ and from $v$ to $u$. In addition $S$ is maximal in the sense that adding any other set of vertices and/or 
edges from $G$ to $S$ would break this property.
\end{definition}
\begin{definition}\label{def:CAC}
A connected acyclic component (CAC) of a directed graph $G$ is a subgraph $S$ of $G$ such that no vertex in $S$ is part of any non-loop cycle in $G$
and the underlying graph is connected. Additionally any edge in $G$ that exists between any two vertices in $S$ is also a part of $S$. 
A vertex in the CAC with no edge to any other edge in the CAC we call a leaf of the CAC.    
\end{definition}
CACs can be seen as a connected collection of 1-vertex SCCs forming a tree. While CACs keep the property that all internal edges between vertices in the component are preserved from those in the original graph, it is not maximal in the sense that no more vertices could be added to the component as is the case for SCCs. The reason for this is that we want to be able to create a graph partitioning into components in which the underlying graph is a directed acyclic graph (DAG) in the same way as for the ordinary partition into SCCs. 
\begin{definition}\label{def:Level}
Consider a graph $G$ with partition $P$ into SCCs and CACs such that each vertex is part of exactly one component and the underlying graph created by replacing every component with a single vertex. If there is an edge between any two vertices between a pair of components then there is an edge in the same direction between the two vertices representing those two components as well. Consider the case where the underlying graph is a DAG (such as for the commonly known partitioning of a graph into SCCs).
\begin{itemize}
\item The level $L_C$ of component C is equal to the length of the longest path in the underlying DAG starting in C.  
\item The level $L_{v_i}$ of some vertex $v_i$ is defined as the level of the component for which $v_i$ belongs ($L_{v_i} \equiv L_C, ~\text{if } v_i \in C$). 
\end{itemize}
\end{definition}
We note that a SCC made up of only a single vertex is also a CAC, in our work it will be easier to consider these components CACs rather than SCCs. We also note that while a single vertex can only be part of a single SCC, it could be part of multiple CACs of different size. There is however a unique graph partitioning into SCCs and CACs as seen below. 
\begin{theorem}
Consider a directed graph $G$ with a partition into SCCs. Let the underlying graph be the DAG constructed by replacing every component with a single vertex. If there 
is an edge between any two vertices between a pair of components then there is an edge in the same direction between corresponding vertices in the underlying DAG. To each vertex in the underlying DAG we attach a level equal to the longest existing path from this vertex to any other vertex in the underlying graph. 
Next we start to merge SCCs consisting of a single vertex into CACs under the following conditions:
\begin{itemize}
\item We start merging from the lowest level (vertices in the DAG with no edge to any other vertex in the DAG) and only start merging on the next level when we cannot merge any more components on the current level.
\item All merges are done by merging a single `head' 1-vertex CAC of level $L$ containing vertex $v$ with all CACs of level $L-1$ to which there is an edge from $v$. Unless $v$ have an edge to at least one SCC (of more than 1 vertex) of level $L-1$ in which case no merge is made. If a merge takes place, then the level of the new merged CAC is $L-1$. 
\end{itemize}
Then the following holds:
\begin{enumerate}
\item This gives a unique partitioning of the graph into SCCs and CACs and does not depend on the order in which we apply merges of `head' components on the same level.
\item This partition of SCCs and CACs can also be seen as a DAG where we attach a level to each vertex equal to the longest existing path from this vertex to any other vertex in the DAG.
\end{enumerate}
\end{theorem}
\begin{remark}
Note that after a merge some vertices with a level higher than the one where the merge was made might get a lower level compared to before. 
\end{remark}
\begin{proof}
That a directed graph can be partitioned into SCCs is a well-known and easy to show result from graph theory. Applying a level to the vertices in the DAG is nothing else than a topological ordering of the vertices in the graph, something also well-known, hence we start at the merging. Obviously all 1-vertex SCCs are also 1-vertex CACs since any vertex that is part of any (non-loop) cycle must be part of a SCC of more than one vertex. 

Since the head CAC of a merge is always connected with each other CAC that is part of a merge, the subgraph representing the component is connected as well. It is also still obviously acyclic since no vertex of any of the CACs is part of any cycle in $G$ from the definition of a CAC. Adding all edges between the head and all other merged CACs also ensures that there is no missing edge between any two vertices of the new CAC. This holds since there can be no edges between any two components on the same level. Hence we can conclude that merging CACs creates a new CAC. 

Next we prove statement 1) that the given partitioning is unique and does not depend on the order of merges.
CACs are created using a bottom-up approach and it is clear that the level of a CAC never change after its first merge. This means that the level of a CAC is uniquely defined by the level of any leaf in the CAC. All the leaves of a CAC are those 1-vertex CACs which could not be the head of any merge either because they have either no outgoing edges orthey have at least one outgoing edge to a SCC (of more than 1 vertex) of the next lower level.

Assume we have done all merges with head CACs of level $L$. Consider a 1-vertex CAC with vertex $v$ and level $L+1$ and edges to one or more CACs but no SCC of more than one vertex of level $L$ (or higher). From the previous argument we get that the level of any CAC linked to by $v$ will not change from any future merges. This means that eventually
$v$ will be part of the same CAC as all vertices part of any CAC with level $L$ to which there is an edge from $v$ regardless of which order merges are made on level $L+1$. 

Repeating this argument for all 1-vertex CACs of level $L+1$ we get that for each such vertex $v$ the neighboring vertices part of a CAC with a lower level
for which $v$ should be part of the same CAC is uniquely determined after all the merges on level $L$. Repeating this for all levels gives us a unique partitioning. This proves statement 1). 

Last we show statement 2) by showing that merging of components does not create any cycle in the underlying DAG created by the components. 

Consider a merge of head CAC with vertex $v$, level $L$ and an edge to each of $n$ CACs $C_1,C_2,\ldots,C_n$ with level $L-1$. Since all CACs $C_1,C_2,\ldots,C_n$ have the same level, the resulting CAC after merge can only have edges to components of level at most $L-2$ since we merge it with all CACs of level $L-1$ to which there is an edge from $v$, but do not merge if there is an edge from $v$ to any SCC of level $L-1$ and there can be no edges between any components of the same level from the initial SCC partitioning. Since initially there can be no edges from any component to another of the same or larger level and we do not create any such edges when doing a merge, we do not create any cycle in the underlying DAG. 

Since we only merge CACs, and merges does not create any cycles in the DAG, we have proved that the new partitioning can also be represented by a DAG where each vertex corresponds to one SCC or CAC. This proves statement 2). 
\end{proof}
\begin{corollary}
If a directed graph $G$ has a SCC partitioning with maximum level $L_{SCC}$ and a SCC/CAC partitioning with maximum level $L_{SCC/CAC}$, then
$$L_{SCC/CAC} \le L_{SCC}$$ 
\end{corollary}
\begin{proof}
Every merge lowers the level of the `head' vertex and possibly one or more components of a higher level, this makes it easy to find a graph such that $L_{SCC/CAC} < L_{SCC}$ such as the $2$-vertex graph with a single edge between them in one direction. Similarly we can easily find a graph for which equality holds (such as the single vertex graph). However, since doing a merge can never result in any vertex getting a higher level, $L_{SCC/CAC} \le L_{SCC}$. 
\end{proof}
For algorithms (such as PageRank) which can be done on the components of one level at a time in parallel the SCC/CAC partitioning have the advantage over the usual SCC
partitioning in that it generally creates a lower number of levels and thus a larger amount of vertices (but not necessary components) on the same level in the new partitioning on average. 
In effect reducing the chance of bottlenecks where we have a level with only a single or a few small components. The only components that get larger are the CACs which are acyclic apart from any loops and thus often have specialized faster methods (for example by exploiting the triangular adjacency matrix).  thus increasing the size of these components is often not as much of an issue and might in fact often be beneficial instead by reducing overhead.

An example of a directed acyclic graph and its SCC/CAC partitioning can be seen in Fig. \ref{fig:partition}.
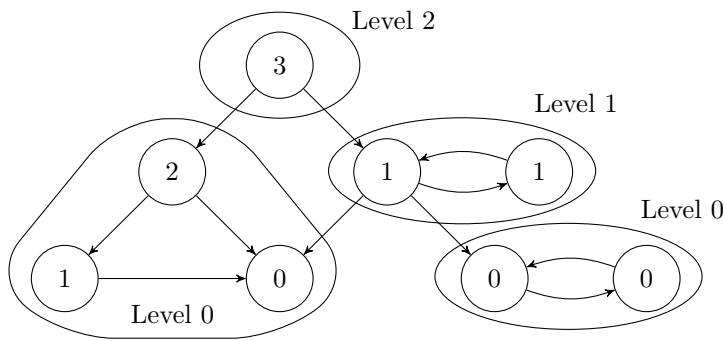
\begin{figure} [!hbt]
\captionsetup{width=0.8\textwidth}
\begin{center}
 \begin{tikzpicture}[->,shorten >=0pt,auto,node distance=2.0cm,on grid,>=stealth',
every state/.style={circle,,draw=black}]
\node[state] (A) [] {3};
\node[state] (B) [below right=of A] {1};
\node[state] (C) [below left=of A] {2};
\node[state] (D) [right=of B] {1};
\node[state] (E) [below left=of C] {1};
\node[state] (F) [below right=of C] {0};
\node[state] (G) [below right=of B] {0};
\node[state] (H) [right=of G] {0};
\node at (1.5,0.6) {Level 2};
\node at (3.9,-.5) {Level 1};
\node at (5.3,-1.9) {Level 0};
\node at (-1.41,-3.3) {Level 0};
\path (A) edge node {} (B)
(A) edge node {} (C)
(B) edge [bend right=20] node {} (D)
(C) edge node {} (E)
(C) edge node {} (F)
(B) edge node {} (F)
(E) edge node {} (F)
(G) edge [bend right = 20] node {} (H)
(D) edge [bend right = 20] node {} (B)
(H) edge [bend right = 20] node {} (G)
(B) edge node {} (G);
\draw (2.4,-1.4) ellipse (50pt and 20pt);
\draw (3.8,-2.8) ellipse (50pt and 20pt);
\draw (0,0) ellipse (30pt and 20pt);
\draw [rounded corners= 18pt] (-1.41+0.7,-1.41+0.7) -- (0+1,-2.82) -- (0,-2.82-0.8) 
-- (-2.82,-2.82-0.8) -- (-2.82-1,-2.82) -- (-1.41-0.7,-1.41+0.7)--cycle;
\end{tikzpicture}
\end{center}
\caption{Example of a graph and corresponding components from SCC/CAC partitioning of the graph (2 SCCs, 1 CAC and 1 1-vertex component). Vertex labels denote the level of each vertex if we had only partitioned the graph into SCCs (for the SCC/CAC partitioning the vertex-levels is the same as the level of corresponding component in the figure).}
\label{fig:partition}
\end{figure}
From this figure it is clear why we cannot merge when the `head' have an edge to any SCC of the next lower level. If the top (level 2) component merged with the left (level 0) component then this would have created a cycle in the underlying graph. It is also possible to see how merging some components can result in the partitioning getting a lower max-level, the SCC/CAC partitioning have only 3 levels while the SCC partitioning would need 4 levels.
\subsection{PageRank}\label{sec:PageRank}
PageRank was originally defined by S. Brin and L. Page as the eigenvector to the dominant eigenvalue of a modified version of the adjacency matrix of a graph \cite{tAnatomy-largeWebSearch}.
\begin{definition}\label{def:R1}
PageRank $\vec{R}^{(1)}$ for vertices in graph $G$ consisting of $|V|$ vertices is defined as the (right) eigenvector with eigenvalue one to the matrix: 
\begin{equation}
{\bf M} = c({\bf A}+\vec{g}\vec{w}^\top)^\top+(1-c)\vec{w}\vec{e}^\top
\end{equation}
where ${\bf A}$ is the adjacency matrix weighted such that the sum over every non-zero row is equal to one (size  $|V| \times |V|$), $\vec{g}$ is a $|V| \times 1$ vector with zeros for vertices with outgoing edges and $1$ for all vertices with no outgoing edges, $\vec{w}$ is a  $|V| \times 1$ non-negative vector with $||\vec{w}||_1 = 1$, $\vec{e}$ is a one-vector with size $|V| \times 1$ and $0 < c < 1$ is a scalar. 
\end{definition}
The original normalized version of PageRank has the disadvantage in that it is harder to compare PageRank between graphs or components, because of that we use a non-normalized version of PageRank as described in for example \cite{ce_SMTDA2014}.
\begin{definition}\label{def:R3}
{\rm (\cite{ce_SMTDA2014})} PageRank $\vec{R}^{(3)}$ for graph $G$ is defined as
\begin{equation}
\vec{R}^{(3)} = \frac{\vec{R}^{(1)}||\vec{W}||_1}{d}, ~d = 1 - \sum{c{\bf A}^\top} \vec{R}^{(1)}
\end{equation}
where $\vec{W}$ is a non-negative weight vector such that $\vec{W} \propto \vec{w}$.
\end{definition}
In the same work \cite{ce_SMTDA2014} we also showed that it is possible to give another equivalent definition of this non-normalized version of PageRank which will be useful later in some proofs. 
\begin{definition}\label{def:R3prob}
{\rm (\cite{ce_SMTDA2014})} Consider a random walk on a graph $G = \{V,E \}$ described by ${\bf A}$. In each step of the random walk move to a new vertex from the current vertex by traversing a random edge from the current vertex with probability $0<c<1$ and stop the random walk with probability $1-c$. Then PageRank $\vec{R}^{(3)}$ for a single vertex $v_j$ can be written as
\begin{equation} \label{eq:R3prob}
R^{(3)}_j = \left(W_j+\sum_{v_i \in V,v_i\neq v_j}{W_iP(v_i \rightarrow v_j)}\right) \left(\sum_{k=0}^{\infty}{(P(v_j \rightarrow v_j))^k}\right)
\end{equation}
where $P(v_i \rightarrow v_j)$ is the probability to hit vertex $v_j$ in a random walk starting at vertex $v_i$. This can be seen as the expected number of visits to $v_j$ if we do multiple random walks, starting at every vertex a number of times described by $\vec{W}$.
\end{definition}
In \cite{ce_ASMDA2015} we showed how to calculate PageRank for the five different types of vertices defined below
\begin{definition}
For the vertices of a simple directed graph we can define $5$ distinct groups $G_1,G_2,\ldots,G_5$ 
\begin{enumerate}
\item $G_1$: Vertices with no outgoing or incoming edges.
\item $G_2$: Vertices with no outgoing edges and at least one incoming edge (also called dangling vertices).
\item $G_3$: Vertices with at least one outgoing edge, but no incoming edges (also called root vertices). 
\item $G_4$: Vertices with at least one outgoing and incoming edge, but which is not part of any (non-loop) directed cycle (no path from the vertex back to itself apart from the possibility of a loop).
\item $G_5$: Vertices that is part of at least one non-loop directed cycle.
\end{enumerate}
\end{definition}
Qing Yu et al  gave a similar but slightly different definition of 5 (non distinct) groups for vertices, namely dangling and root vertices ($G_2$ and $G_3$), vertices that can be made into dangling or root vertices by recursively removing dangling or root vertices (part of $G_4$) and remaining vertices (part of $G_4$ and $G_5$) \cite{Qing_lumping2}. 
Given PageRank of a vertex not part of a cycle (group 1-4), then the PageRank of other vertices can be calculated by removing the vertex and modifying the initial weight of other vertices. 
\begin{theorem}\label{thm:vertex}
Given PageRank $\vec{R}^{(3)}_{g} $ of vertex $v_g$  where $v_g$ is not part of any non-loop cycle, the PageRank of another vertex $v_i$ from which there exist no path to $v_g$ can be expressed as
\begin{equation}
R^{(3)}_{i} = \left(W_i+R^{(3)}_{g}ca_{gi}+ \sum_{\stackrel{v_j \in V}{ v_j \neq v_i,v_g}}{(W_j+R^{(3)}_{g}ca_{gj}) P(v_j\rightarrow v_i)}\right) \left(\sum_{k=0}^{\infty}{(P(v_i \rightarrow v_i))^k}\right) 
 \end{equation}
where $ca_{gi}$ is the one-step probability to go from $v_g$ to $v_i$. 
\end{theorem}
The proof with minor modifications is similar to the one found in \cite{ce_ASMDA2015} where it is formulated for vertices in $G_3$ on graphs with no loops.
\begin{proof}
Consider $R^{(3)}_{g}$ from Definition. \ref{def:R3prob}. 
Since we know that there is no path from $v_i$ back to $v_g$ (or $v_g$ would be part of a non-loop cycle) we know that the right hand side will be identical for all other vertices. We rewrite the influence of $v_g$ using
\begin{equation}\label{eq:G3_one_step}
R^{(3)}_{g}P(v_g \rightarrow v_i) = R^{(3)}_{g}ca_{gi} +  \sum_{\stackrel{v_j \in V}{ v_j \neq v_i,v_g}}{R^{(3)}_{g}ca_{gj} P(v_j\rightarrow v_i)} \enspace .
\end{equation}
We can now rewrite the left sum in Definition. \ref{def:R3prob}:
\begin{equation}
\sum_{v_i \in V,v_i\neq v_j}{W_iP(v_i \rightarrow v_j)} = R^{(3)}_{g}ca_{gi}+ \sum_{\stackrel{v_j \in V}{ v_j \neq v_i,v_g}}{(W_j+R^{(3)}_{g}ca_{gj}) P(v_j\rightarrow v_i)} \enspace 
\end{equation}
which when substituted into (\ref{eq:R3prob}) proves the theorem.
\end{proof}

It is also easy to show that any SCC can also be divided into one of the first four groups if we consider each SCC as a vertex in the underlying DAG (a SCC can never be part of a cycle). The important part of this is that it is also possible to calculate PageRank one component at a time rather than for the whole graph at once. 

\begin{corollary}\label{cor:level}
Let $\vec{R}^{(3)}_{L^+}$ be PageRank of all vertices belonging to components of level $L$  or greater. Then PageRank of a vertex $v_i$ belonging to a component of level $L-1$ can be computed by
\begin{gather*} R^{(3)}_{i} = \left(\sum_{k=0}^{\infty}{(P(v_i \rightarrow v_i))^k}\right) \\
\left( \left(W_i+c\left(\vec{R}^{(3)}_{L^+}\right)^\top \vec{a}_{L^+,i}\right)+ \sum_{\stackrel{v_j \in V}{ v_j \neq v_i}}{\left(W_j+c\left(\vec{R}^{(3)}_{L^+}\right)^\top \vec{a}_{L^+,j}\right) P(v_j\rightarrow v_i)}\right) \end{gather*}
where $\vec{a}_{L^+,i}$ is a vector containing all 1-step probabilities from vertices of level $L$ or greater to vertex $v_i$.
\end{corollary}
\begin{proof}
Follows immediately from Theorem \ref{thm:vertex} by replacing the rank of a single vertex with the sum of rank of all vertices belonging to components of a higher level. Those in lower level components or other components on the same level does not affect the rank since they automatically does not have any path to $v_i$.
\end{proof}

Using Corollary \ref{cor:level} it is clear that after calculating PageRank of all vertices belonging to components of level $L$ and above we can calculate those of level $L-1$ by first changing their initial weight and then consider the component by itself. In matrix notation we can update the weight vector for all components of lower level by calculating

\begin{equation}\label{PR:levels}
\vec{W}^{new}_{L-1} = \vec{W}^{old}_{L-1}+c{\bf A}_{L^+,L-1}\vec{R}^{(3)}_{L^+}
\end{equation}
where ${\bf A}_{L^+,L-1}$ corresponds to the submatrix of $A$ with all rows corresponding to vertices of level $L$ or greater and all columns of level $L-1$. 
This is essentially the same method which is used in \cite{journals/im/LeeGZ07} but here we have formulated it for any component instead of for dangling vertices (vertices with no outgoing edges). 

\section{Method}\label{sec:method}
The complete PageRank algorithm can be described in three main steps.
\begin{enumerate}
\item Component finding: Finding the SCC/CAC partitioning of the graph.
\item Intermediate step: Create relevant component matrices and weight vectors. 
\item PageRank step: Calculate PageRank one level at a time and components on the same level one at a time or in parallel.  
\end{enumerate}
In order to be able calculate PageRank for each component we obviously first need to find the components themselves, this is done in the component finding part of the algorithm where we find a SCC/CAC partitioning of the graph as well as the level of each component. By using the CAC/SCC partitioning rather than the usual SCC partitioning we reduce the risk of having very few vertices on the same level, the aim of this is to be able to avoid some of the disadvantages with some other similar methods such as the one in \cite{Langville:2005:RPP:1114107.1117871} where a large number of small levels (small diagonal blocks) increases the overhead cost \cite{Qing_lumping2}. This step is similar to the initial matrix reordering made by \cite{Arasu_et_at}. However instead of only finding a partial ordering we have modified the depth first search slightly in order to identify components that can be calculated in parallell as well as group 1-vertex components on the same level together. Another advantage is that different methods can be used for different types of components as we will see later. The component finding step is described in Sec. \ref{sec:component_finding}. 

In the intermediate step the data (edge list, vertex weights) need to be managed such that the individual matrices for every component can quickly and easily be extracted. This section of the code can vary a lot between implementations and is one of the main contributors of overhead in the algorithm. The SCC/CAC partitioning can easily be transformed into a permutation matrix and used to permute the graph matrix and then solve the resulting linear system, this can be seen as an alternative to the recursive reordering algorithm described in \cite{Langville:2005:RPP:1114107.1117871}. This step is described in Sec. \ref{sec:intermediate}.

After the intermediate step we are ready to start calculating PageRank of the vertices. This is done one level at a time starting with the highest and modifying vertex weights between levels using (\ref{PR:levels}). Components on the same level can use different methods to calculate PageRank and can either be calculated sequentially or in parallel. The PageRank step is described in Sec. \ref{sec:PageRank_step}.

\subsection{Component finding}\label{sec:component_finding}
The component finding part of the algorithm consists of finding a SCC/CAC partitioning of the graph as well as the corresponding levels of the components. Since any loops in the graph have no effect on which SCC or CAC a vertex is part of, these will be ignored in the component finding step. Finding the components and their level can be done through a modified version of Tarjan's well-known SCC finding algorithm using a depth first search \cite{Tarjan72depthfirst}. For every vertex $v$ we assign six values:
\begin{itemize}
\item v.index containing the order in which it was discovered in the depth first search.
\item v.lowlink for a SCC representing the lowest index of any vertex we can reach from $v$, or for a CAC representing the `head' vertex of corresponding component.
\item v.comp representing the component the vertex is part of, assigned at the end of the component finding step. 
\item v.depth used to implement efficient merges of components. It can be removed if the extra memory is needed, but it could result in slowdown because of merges for some graphs. 
\item v.type indicates if $v$ is part of a SCC or a CAC (1-vertex SCCs are considered CACs). 
\item v.level indicating the level of the component to which $v$ belongs. 
\end{itemize}
Of these the first three can be seen in Tarjan's algorithm as well, and play virtually the same role here (although in Tarjan's the comp value can be assigned as components are created).
During the depth first search each vertex $v$ goes through three steps in order.
\begin{enumerate}
\item Discover: Initialize values for the vertex.
\item Explore: Visit all neighbors of $v$, finishing the DFS of any unvisited neighbors before going to the next. After a vertex is visited we update v.lowlink and v.level. 
\item Finish: After all neighbors are visited we create a new component if appropriate. If a CAC is created we also check for and do any merge with $v$ as head.
\end{enumerate}
During the discover step values are initialized after which the vertex is put on the stack (.type and .comp do not need to be initialized)
\begin{lstlisting}[frame=single] 
Discover(Vertex v) 
	v.index := index
	v.lowlink := index
	v.level := 1
	v.depth := 1
	index++
	stack.push(v) //add v to the stack
end
\end{lstlisting}
Here index is a counter starting at $1$ and increasing by one for every new vertex we discover. The explore step as well works much like Tarjan's depth first search except that it also updates the level of the vertex we are exploring. 
\begin{lstlisting}[frame=single] 
Explore(Vertex v)
	for each (v, w) %* $\in$ *)  E
		if w is not initialized
			DFS(w) //Discover(w), Explore(w) and Finish(w)
		end
		/* At this point w is either in a component already 
		or belong to the same SCC as v */
		if w is in a component (.type is defined)
			v.level = max(v.level,w.level+1)
		else 	//w belong to the same SCC
			v.level = max(v.level,w.level)
			v.lowlink = min(v.lowlink,w.lowlink)
		end
	end
end	
\end{lstlisting}
The last step in the DFS is where we evaluate if a new component should be created and handle any merges needed with this vertex as `head' component. The initial component is created in the same way as in Tarjan's algoriithm: if v.lowlink $=$ v.index we create a SCC by popping vertices from the stack until we pop $v$ from the stack.

\begin{lstlisting}[frame=single] 
Finish(Vertex v)
	if v.lowlink = v.index
		size = 0;
		do
			w := stack.pop() 
			w.lowlink = v
			w.level = v.level
			w.type = scc
			merge(w,v) //add w to component
			size++
		while (w != v)
		
		if size = 1 //(one vertex component)
			v.type = cac
			//check for merges
			for each (v, w) %* $\in$ *) E
				m = false
				list = []
				if w.type = scc and w.level = v.level-1
					m = false
					break
				end
				if w.type = cac and w.level = v.level-1 
					list.add(w)
					m = true
				end
			end
			if m = true //(adjust .level if a merge occured )
				v.level = v.level-1
				for each w in list
					merge(w,v) 
				end
			end
		end	
	end
end
\end{lstlisting}
The first part is similar to Tarjan's with some extra bookkeeping so we will focus our explanation of the second part. However in order to explain how the merges are done
in constant time we start by explaining how we store the component data. The components are stored in a merge-find data structure through the .lowlink and .depth attribute. 
A merge-find data structure allows us to do the two operations we need: merging two components and finding a `head' vertex representing a component (used in merge, and by itself later).
This has the advantage that both operations can be done in constant amortized time ($O(\alpha(|V|))$) as well as requiring very little memory.

In Tarjan's algorithm you don't need the .depth values since the .comp value can be assigned while creating a component. The reason we do not do it here is because it cannot be updated when doing a merge of two CACs (unless we loop through all vertices), and even if we could, it would be possible to end up with some empty components that would have to be 'cleaned up' in one way or another later anyway. 

Looking at the computational complexity of the component finding step we see that discover, explore and finish are all done once for every vertex, of these
discover is obviously done in constant time. During explore we will eventually have to go through all edges exactly once but all operations take only
constant time. Last finish is called once for every vertex doing $O(\alpha(|V|))$ work in the first half (amortized constant because we do merges rather
than assigning component values directly). In the second half we will at most visit every edge once (over all vertices) doing $O(\alpha(|V|))$ work. 
Thus in total we end up with $O(|V|+|E|+|V|\alpha(|V|)+|E|\alpha(|V|)) \approx O(|E|\alpha(|V|)), \text{ if } |E| > |V| $ , in other words linear amortized time in the number of edges.    

Before returning the results $|V|$ find operations also need to be done in order to assign the .comp value for each vertex, since the find operation takes $O(\alpha (|V|))$ time this takes $O(|V|\alpha(|V|))$ time in total.
\begin{lstlisting}[frame=single] 
for each v %* $\in$ *) V
	h := find(v) //find head vertex
	ind := 1
	if h.comp is not defined
		h.comp := ind
		ind++
	end
	v.comp := h.comp	
end
\end{lstlisting}

Hence the complete component finding algorithm takes $O(|E|\alpha(n))$ time which is comparable to Tarjan's which can be implemented in $O(|E|)$ time. We note that if the .depth value is ignored everything works but the merges are no longer guaranteed to be made in constant amortized time. If memory is a concern or if the size of the CACs are assumed to be small it might be worthwile to work without the .depth value even though merges could be slow in the worst case scenario. 

\subsection{Intermediate step}\label{sec:intermediate}
After we have found the SCC/CAC partitioning some additional work need to be done in order to continue with the PageRank calculations effectively. 

The goal of the preprocessing step is to make sure that we easily and quickly can construct corresponding matrix for each component as well 
as sort the components. We sorted components first in order of level and second in the size of the component (both descending order) so that we can work on one
component at a time starting with the largest on every level.

We note that the preprocessing step can differ highly depending on implementation, hence we will only give a short comment on how we choose to do it.
The goal is to have separate edge lists for edges within each component as well as edges between levels. To get this we start by sorting the
first according to their level and second according to their size. This is then used to permute the edge list such that they are ordered in the same way depending on their starting vertex. It is important to note that this sorting of components or edges can both be done in linear time rather than $O(n\log{n})$ as with ordinary sort. This is true since the possible values are known and bounded hence no comparisons  need to be made.   

After the edges are sorted we go through the edge list and assign each edge to the correct component or in between level list such that corresponding 
component matrices can be created. In this step we also merge all 1-vertex components of each level to avoid having to calculate the rank of multiple 1-vertex components of the same level.

After the preprocessing step we should have the matrices (or edge lists) $M_c$ and weight vector $V_c$ for all components stored such that we can access them when needed.

It is worth to note that in our implementation this section of the code contains much of the extra overhead needed for our method (more than the previous component finding step itself).  

\subsection{PageRank step}\label{sec:PageRank_step}
Now that all preliminary work is done we can start the actual PageRank calculation where we calculate PageRank for all vertices of one level at a time (starting by the largest). The PageRank step can be described by the following steps. 

\begin{enumerate}
\item Initiate $L$ to the maximum level among all components. 
\item For each component of level $L$: pick a suitable method and calculate PageRank of the component. 
\item Update weight vector $V$ for all remaining components (of lower level).
\item Decrease $L$ by one and go to step $2$ unless we have already calculating PageRank of all vertices. 		
\end{enumerate}
Depending on the type and size of the component PageRank we calculate PageRank in one of four different ways: 

\begin{itemize}
\item Component is made up of a collection of single vertex components (no internal edges apart from loops): PageRank of any such collection of 1-vertex components is simply the initial weight $w_i$ for vertices with no loop and $\frac{w_i}{1-ca_{ii}}$ for any vertex with a loop, where $a_{ii}$ is the weight on the loop edge.  
\item Component is a CAC of more than one vertex: use CAC-PageRank algorithm \ref{sec:CACalg} described later.
\item Component is a SCC but small (for example less than 100 vertices): Calculate PageRank by directly solving the linear equation system $({\bf I}-c{\bf A}^\top)\vec{R} = \vec{W}$ (using for example LU factorization).   
\item Component is a SCC and large: Use iterative method, in our case using a power series formulation.  
\end{itemize}
Out of these the first one is done in $O(|V|)$ time (copy the weight vector) or $O(1)$ if no separate vector is used for the resulting PageRank and we assume there is no loops, the second and fourth is done in $O(|E|)$, however the coefficient in front of the first is much lower since it is guaranteed to only visit every edge once, while the iterative method needs to visit every edge in every iteration (number of which depend on error tolerance and method chosen). The third is done in $O(|V|^3)$ using LU factorization, however since $|V|$ is small this is still faster than the iterative method unless the error tolerance chosen is large.  

We also note that only the fourth method actually depend on the error tolerance at all, every other method can be done in the same time regardless of error tolerance (down to machine precision). \\

After we have calculated PageRank for all components on the current level we need to adjust the weight of all vertices in lower level components as shown in \ref{PR:levels}. This can be done using a single matrix-vector multiplication using the edges between the two sets of components. 
$$ \vec{W}^{\text{new}}_{(L-1)^-} = \vec{W}^{\text{old}}_{(L-1)^-} + {\bf M}^\top_{L^+,(L-1)^-}\vec{R}_{L}$$
This is the same kind of correction as is done in for example \cite{FAndersson_art_PR} and for the non-normalized PageRank used here in \cite{ce_ASMDA2015}. \\

In the weight adjustment step every edge is needed once if it is an edge between components and never if it is an edge within a component.
Hence if we look over the whole algorithm: every edge is visited at most twice in the DFS, then every edge that is not part of a SCC is visited exactly once
more (either as part of a CAC or as an edge between components) while those that are part of a SCC are typically visited a significantly larger number of times
depending on algorithm, error tolerance and convergence criterion used. Of course there is also some extra overhead that would need to be taken into consideration for a more in-depth analysis. \\

We note that calculating PageRank for all components on a single level can be computed in parallel (hence why we sort them by their size starting by the largest).
The weight adjustment can either be done in parallel for each component or as we have done here once for all components of the same level. In case there is a single very large component on a level it might be more appropriate to do it one component at a time instead to reduce the time waiting for the large component to finish. 

\subsubsection{CAC-PageRank algorithm}\label{sec:CACalg}
The CAC-PageRank algorithm exploits the fact that there are no non-loop cycles to calculate the PageRank of the graph using a variation of Breadth first search. 
The algorithm starts by calculating the in-degree (ignoring any loops) of every vertex and stores it in v.degree for each vertex v, this can be done by looping through all edges once.
We also keep a value v.rank initialized to corresponding value in the weight vector for each vertex. The BFS itself can be described by the following 
psuedocode. 
\begin{lstlisting}[frame=single] 
for each v %* $\in$ *) V
	if v.degree = 0	
		Queue.enqueue(v)
		while Queue.size > 0
			w = Queue.dequeue()
			w.rank = w.rank*(1/(1-W(w,w)))	 //adjust for loops
			for each (w,u) %* $\in$ *) E/(w,w)
				u.rank = u.rank + w.rank * W(w,u)
				u.degree = u.degree-1
				if w.degree = 0
					Queue.enqueue(w)
				end
			end
			w.degree = w.degree-1  //ensure w is never enqued again
		end
	end
end
\end{lstlisting}
Here W(w,u) is the weight on the (w,u)-edge ($ca_{wu}$). Note that $W(w,w) = 0$ if $w$ has no loop, which gives simply multiplication by 1 when adjusting for loops. The difference between this and ordinary BFS is that we only add a vertex $v$ to the queue once we have visited all incoming edges to $v$. We also loop through all vertices to ensure that we wisit each vertex once since there could be multiple vertices with no incoming edges. 

Usually PageRank is defined only for graphs with no loops (or by ignoring any loops that are present), this simplifies the algorithm slightly in that the loop adjustment step can be ignored. 

Looking at the computational complexity it is easy to see that we visit every vertex and every edge once doing constant time work, hence we have the same time
complexity as for ordinary BFS $O(|V|+|E|) \approx O(|E|), \text{ if } |E| > |V|$. While it has the same computational complexity as most numerical methods used to calculate PageRank, in practice it will often be much faster in that the coefficient in front will be smaller, especially as the error tolerance decreases.

\subsubsection{SCC-PageRank algorithm}\label{sec:SCC:PR}
If the component is a SCC, then we cannot use the previous algorithm. Instead one of many iterative methods needs to be used (unless the size of the component is small). We have chosen to calculate PageRank for SCCs using a power series as described in \cite{FAndersson_art_PR} since it is a more natural fit with the non-normalized variation of PageRank used here. The method works by iterating the following until some convergence criterion is met.

$$\left \{ \begin{array}{ll} \vec{P}_{n+1} &= c{\bf A}\vec{P}_{n}\\
\vec{P}_{0} &= \vec{W} \end{array} \right.$$
$$\vec{R}^{(3)}_{n} = \sum_{k=0}^{n}\vec{P}_k $$

Although any method can be used, by using a power series we have the advantage in that we get the PageRank in the correct non-normalized form we need without the need to re-scale the result. In practice any other method such as a power iteration could be used as well by scaling the result as described in Def. \ref{def:R3}. The calculation of these components could also benefit from the use of other methods designed to improve the calculation time of PageRank such as the adaptive algorithm described in \cite{Kamvar200451}. Any other algorithm which depend on the presence of dangling vertices or SCCs would not be useful here however since we are working on a single SCC.

The method you get by using a power series formulation have been shown to have the same or similar convergence speed as the more conventional Power 
method experimentally \cite{FAndersson_art_PR}. It is easy to show that this is equivalent to using the Jacobi method with initial vector equal to $\vec{W}$, and obviously Gauss-Seidel or a similar method using a power series formulation could be used and would likely provide some additional speedup. Theoretically the convergence of the Power method can be shown to be geometric depending on the two largest eigenvalues with ratio $|\lambda_2|/|\lambda_1|$ where $\lambda_2 \le c$ \cite{ilprints582}, obviously we have a similar convergence by calculating corresponding geometric sum (geometric, bounded from above by $c$). The method can be described through the following pseudocode.

\begin{lstlisting}[frame=single] 
rank = W //initialize rank to the weightvector
mr = W; //(initialize rank added in previous iteration
do
	mr = M*mr
	rank = rank+mr
while (max(mr) %* $\ge$ *) tol)
\end{lstlisting}
This is also the baseline method we use for comparison with our own method later. 
\subsection{Error estimation}\label{sec:error}
When looking at the error we have two different kinds of errors to consider, first errors from the iterative method used to calculate PageRank of large SCCs (depending on error tolerance) 
and second any errors because of errors in data (${\bf M}$ or $\vec{W}$). We will mainly concern ourselves with the first type which is likely to dominate unless the error tolerance is very small. 

We start by looking at a single isolated component, if this component is a CAC or a small SCC we calculate PageRank exactly and errors can be assumed to be small as long as an appropriate method is used to solve the linear system for the small SCCs using for example LU-decomposition. For large SCCs we stop iterations after the maximum change of rank for any vertex between any two iterations is less than
the error tolerance ($tol$). Since the rank is monotonically increasing we can be sure that the true rank is always a little higher in reality than what is actually calculated. 

The true rank can be described by $R^{(3)} = \sum_{k=0}^\infty {\bf M}^k\vec{W}$, where we let $\vec{p}_k = {\bf M}^k\vec{W}$. Since every row sum of ${\bf M}$ is less than or equal to $c$, one has $c |\vec{p}_{k-1}| \ge |\vec{p}_k|$. This means that the maximum change in rank over all vertices in the graph after $K$ iterations is bounded by 
$$\sum_{k=1}^{\infty}c^k|\vec{p}_K| = \frac{c |\vec{p}_K|}{1-c} \approx 5.66|\vec{p}_k|, ~ c=0.85$$

This does not change if there are any edges to or from other components in the graph although the difference will be spread over a larger amount of vertices if there are edges from the component. There might also be additional additive error from components with edges to the single component we are considering. Over all vertices and all components we can estimate bounds for the total error $\epsilon_{tot}$ over all vertices as well as the average error $\epsilon_{avg}$ given some error tolerance $tol$.
$$\epsilon_{tot} < |\text{SCC}_l| \cdot \text{tol} \frac{c}{1-c} $$
$$\epsilon_{avg} < \frac{|\text{SCC}_l|}{|V|} \cdot \text{tol} \frac{c}{(1-c)} \le \text{tol} \frac{c}{(1-c)}$$
where $|SCC_l|$ is the number of vertices part of a 'large' SCC (for which we need to use an iterative method) and $|V|$ is the total number of vertices in the graph. It should be noted that this estimate is likely to be many times larger than in reality unless all the vertices have approximately the same rank. Given that PageRank for many real systems approximately follows a power distribution \cite{Becchetti06thedistribution}, most vertices will have orders of magnitude smaller change in rank when finally those with a very high rank have a change smaller than the error tolerance. Additionally if the graph contains some dangling vertices (vertices with no outgoing edges), then these will further reduce the error (can be seen as an increased chance to stop the random walk). 
\section{Experiments}\label{sec:experiments}
Our Implementation of the algorithm is done in a mixture of c and c++ for the graph search algorithms (component finding and CAC-PageRank algorithm) and Matlab
for the ordinary PageRank algorithm for SCCs and weight adjustment as well as the main code gluing the different parts together. The reason to use c/c++ for some
parts is that while Matlab is rather fast at doing elementary matrix operations (PageRank of SCCs and weight adjustment), it is very slow when you attempt to do for example DFS or BFS on a graph. 
\begin{itemize}
\item Component finding: Implemented in c++ as a variation of the depth first search in the boost library. The method is implemented iteratively
rather than recursively (hence it can handle large graphs which could otherwise give a very large recursion depth). 
\item CAC-PageRank algorithm: Implemented in c. 
\item Power series PageRank algorithm: Implemented in Matlab. Used for SCCs as well as on the whole graph for comparison.
\item Main program: Implemented in Matlab, with c/c++ parts used through mex files. 
\end{itemize}
\subsection{Graph description and generation}
Many real world networks including the graph used for calculating PageRank for pages on the Internet share a number of important properties. 

First of all they should be \emph{scale-free}, these networks are characterized by their degree distributions roughly following a power law, if $k$ is the degree then the cumulative distribution for the degree can be written as $P(k) = k^{-\gamma}$. In practices this means that there are a low number of very high degree vertices as well as a large amount of very low degree vertices. 

Secondly they should be \emph{small-world}, these networks are characterized by two different properties 1) the average shortest path distance between any two vertices in the graph is small, roughly the logarithm of the number of vertices in the graph and 2) the network has a high clustering coefficient. So while the first property implies a high connectivity in the network because of the short distances, the second property says that the network should contain multiple small communities with high connectivity among themselves but low connectivity to vertices outside the own group. 

In order to evaluate the method we have chosen to look at five different graphs of varying properties and size.
\begin{itemize}
\item B-A: A graph generated using the Barab\'asi-Albert graph generation algorithm \cite{Barabasi15101999} with mean degree 12, after generation only some of outgoing edges are kept in order to generate a directed graph. 
This graph contains 1000000 vertices of which 959760 are part of a CAC and 999128 edges, there are 188010 1-vertex CACs out of 239258 CACs in total as well as 19813 SCCs. Maximum component size is 483874 vertices (CAC). A second graph where even fewer outgoing edges was kept was also used.\\

The Barab\'asi-Albert graph generation algorithm creates a graph which is scale-free, however it does not have the small-world property \cite{Estrada:2011:SCN:2181136}, the reason we still used this for one of our tests is that it makes it easy to create a scale-free graph with primary acyclic components. 

\item Web: A graph released by Google as part of a contest in 2002 \cite{google_contest} part of the collection of datasets maintained by the SNAP group at Stanford University \cite{snapnets}.
This graph contains 916428 vertices of which 399605 are part of a CAC and 5105039 edges, there are 302768 1-vertex CACs out of 321098 CACs in total as well as 12874 SCCs. Maximum component size is 434818 vertices (SCC). This graph has both the scale-free and small-world properties making it a good example of the type of graph we would be interested to calculate PageRank on in real applications.\\

We also created two additional even larger graphs by using multiple copies of this graph: 1) the graph composed of ten disjoint copies of this graph and 2) a graph composed of ten copies of the web-graph with a small number (20) extra random edges in order to get a single very large component as is common for real-world networks. 
\end{itemize}
\subsubsection{Barab\'asi-Albert graph generation}
The model works by selecting a starting seed graph of $m_0$ vertices, in our case a $20\times20$ graph with uniformly random edges with mean degree $5$. Then new vertices are added iteratively to the graph one at a time by connecting each new vertex with $m \le m_0$ existing vertices with probability proportional to the number of edges already connected to each old vertex. This can be written as
$$p_i = \frac{d_i}{\sum_j d_j} $$
where $d_i$ is the degree of vertex $i$. 

The Barab\'si-Albert model gives an undirected graph. In order to transform it into a directed graph we then went through each vertex and removed some edges. The number of edges that was kept originating from each vertex $E_{i,keep}$ can be described by
$$E_{i,keep} = \lceil \log_2(E_i) \rceil $$
where $E_i$ is the number of edges originally originating from vertex $i$. We choose $m=12$ which after removal of edges gives an average (in$+$out)-degree of $\log_2(24) \approx 4.6$. After removal of vertices all vertices will have similar out-degree, but the in-degree will be similar to how it was after the original Barab\'si-Albert graph generation (in-degree following a power law). 

\subsection{Results}
All experiments are performed on a computer with a quad core 2.7Ghz(core)-3.5Ghz(turbo) processor (Intel(R) Core(TM) i7-4800MQ) using Matlab R2014a with four threads on any of the parts computed in parallel. 

Three different methods where used
\begin{enumerate}
\item Calculate PageRank as a single large component as described in Sec. \ref{sec:SCC:PR} 
\item Using the method described in \ref{sec:method} with components on the same level calculated sequentially. 
\item Using the method described in \ref{sec:method} with components on the same level calculated in parallel. 
\end{enumerate}
In addition for all three method any loops in the graph where ignored (as is common for PageRank). 

We note that the intermediate step between method 2 and 3 differs. Because of limitations in how the parallelization can be implemented we had to separate edges of the same component into their own cell-array for the parallel version, this accounts for the main difference in overhead between method 2 and 3. It should be noted that the intermediate step is not parallelized (in either version) and is something which could probably be significantly improved by implementation in another programming language. 

After doing the SCC/CAC partitioning of the graph and sorting all components according to their level and component size (both descending order) we can visualize the non-zero values of this new reordered adjacency matrix. The density of non-zeros before and after reordering for the Web graph can be seen in Fig. \ref{fig:spyG}.  
\begin{figure}[!hbt]
\captionsetup{width=0.8\textwidth}
\begin{center}
\includegraphics[width=.85\textwidth]{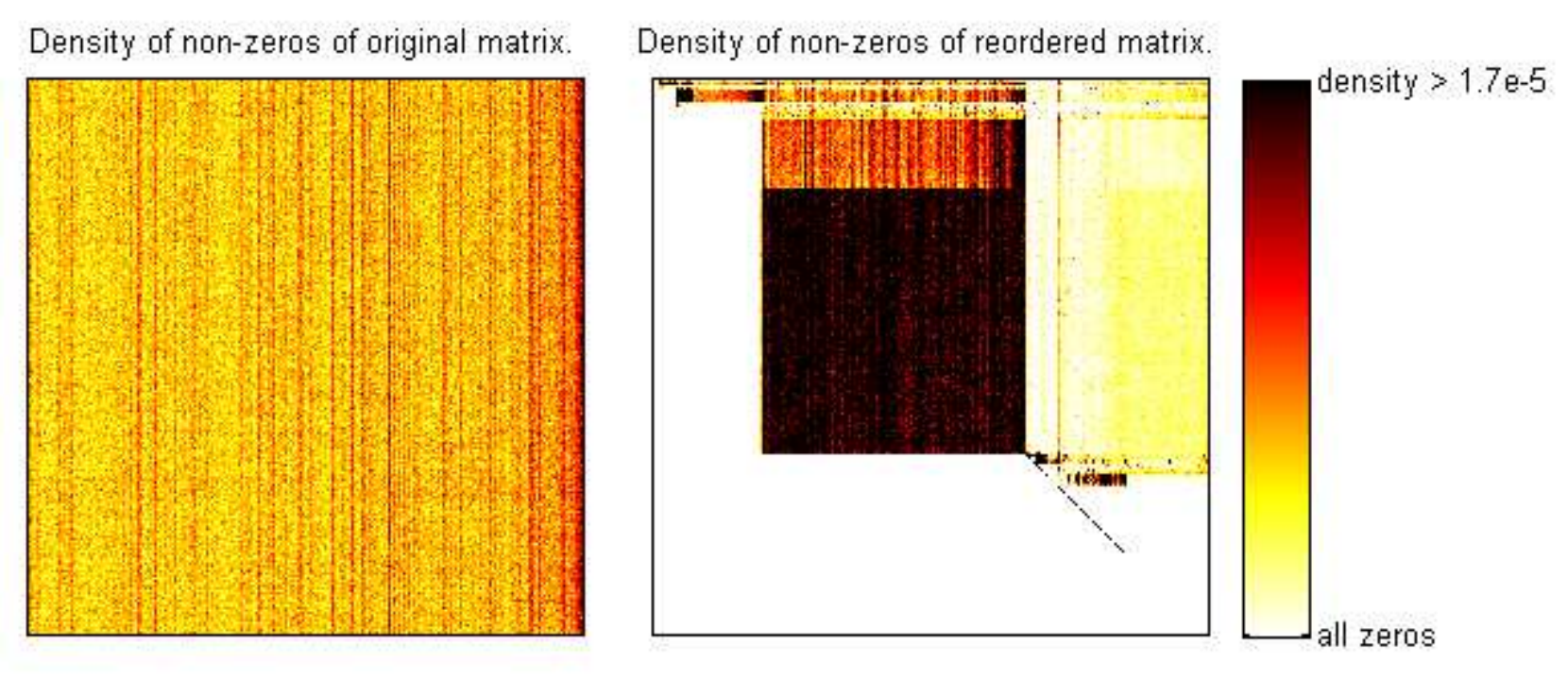}
\end{center}
\caption{Non-zero values of adjacency matrix for the Web graph before and after sorting vertices according to level and component.}
\label{fig:spyG}
\end{figure} 
Note that the diagonal lines are not single vertex components but rather a large amount of small components on the same level (hence they can be computed in parallel), for a view of some of the small components see Fig. \ref{fig:spyG4c}. Any 1-vertex component are not colored since they have no internal edges, two large section of 1-vertex components are right before the middle large component and in the bottom right corner of the matrix. A cutoff at a density of $1.7 \cdot 10^{-15}$ non-zero elements is used in order to avoid problems with a few very high density sections as well as maintaining the same scale in both figures. 

After finding the SCC/CAC partitioning large sections of zeros can clearly be seen, something which is not present in the original matrix. The single very large component in the graph is seen in the middle of the matrix, with a section of small components both above and below it.

Langville and Meyer does a similar reordering by recursively reordering the vertices by putting any dangling vertices last and not considering edges to those already put last once for any further reordering of remaining vertices \cite{Langville:2005:RPP:1114107.1117871}. This effectively creates one or more CACs along with one large component. The advantage of our approach compared to this is that we can also find components above the single large component rather than combining them into a single even larger component as well as finding sets of components which can be computed in parallel. 
\begin{figure}[!hbt]
\captionsetup{width=0.8\textwidth}
\begin{center}
\includegraphics[width=.5\textwidth]{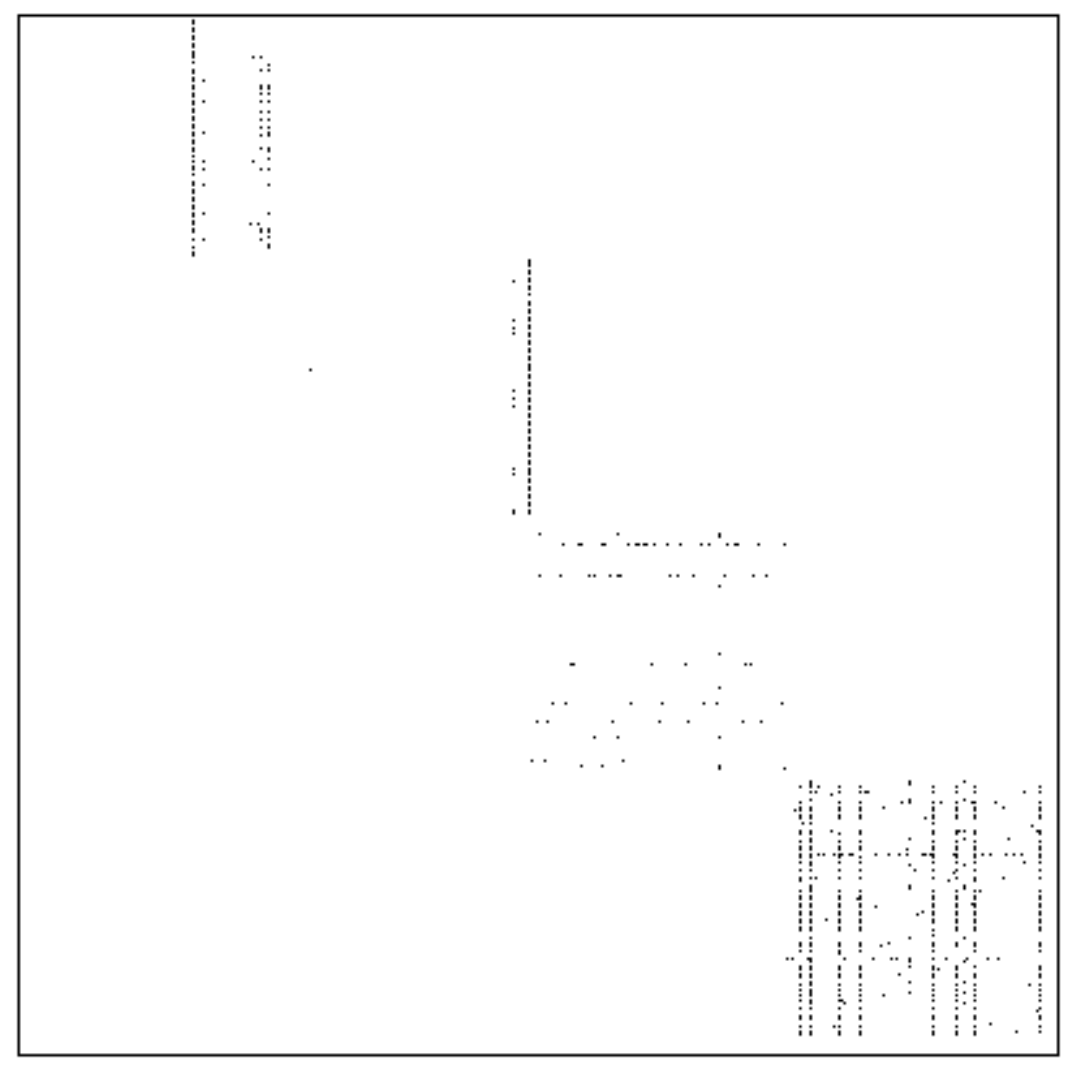}
\end{center}
\caption{Non-zero elements in parts of the bottom right diagonal 'line' of the reordered adjacency matrix.}
\label{fig:spyG4c}
\end{figure} 
In Fig. \ref{fig:spyG4c} some common types of smaller components can be seen. First there are 2 CACs where the majority of the vertices link to the same or a very limited number of vertices creating a very shallow tree. The third component is also a CAC (as seen by the zero rows) have a more advanced structure, while the last one is probably a SCC. There is a large amount of CACs similar to the first 2 with the majority of the vertices linking to one or two vertices but there are also some characterized by a horizontal line representing one or a few vertices linking to a large number of dangling vertices. 

The total number of levels in the Web graph was 28, with the majority being located right at the very top or right after the large central component. More research would be needed to verify if this is usually the case, but if so it might be a good idea to merge some of these very small levels and calculate them as if it was a SCC in order to reduce overhead. For example the first 10 levels contain just 85 vertices in total and could be merged, after the large component there is only 2-3 of these very small components with the rest having at least a hundred or so vertices hence it might not be as useful here. If no merging of 1-vertex CACs was done (using the ordinary SCC partitioning) the number of levels was increased to 34 levels instead.

Since PageRank of different SCCs converge in varying amounts of iterations it is also of interest to see how the number of iterations for different components varies, as well as how it compares to the number of iterations needed by the basic algorithm where we calculate PageRank as if the graph was a single component. Number of iterations for all SCCs of more than $2$ vertices of the Web graph with $c=0.85$ and $\text{tol} = 10^{-9}$ can be seen in Fig. \ref{fig:google_itr}.
\begin{figure}[!hbt]
\captionsetup{width=0.8\textwidth}
\begin{center}
\includegraphics[width=.5\textwidth]{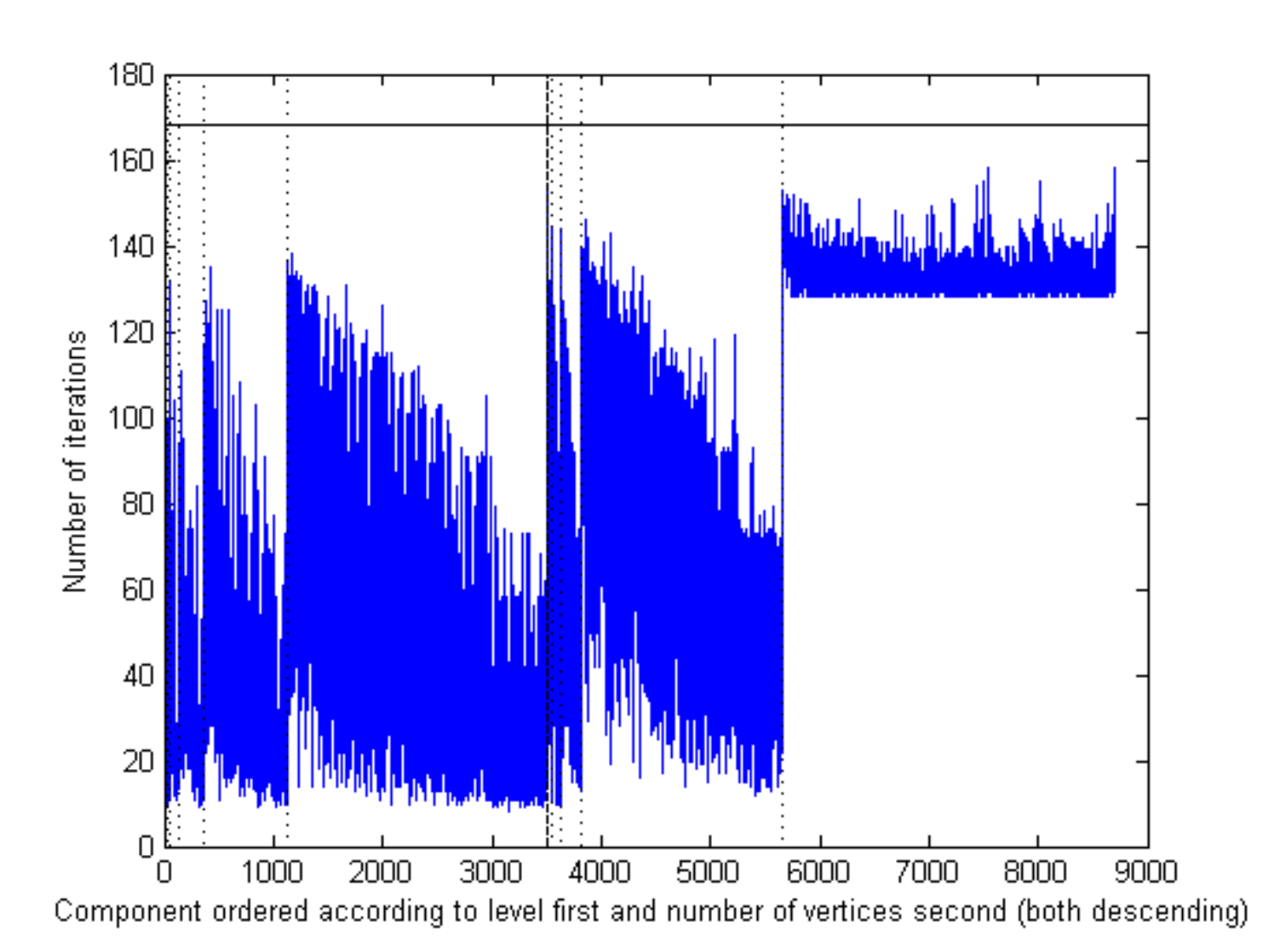}
\end{center}
\caption{Number of iterations needed per SCC ordered according to their level first and number of vertices second (both descending order). The dotted vertical lines denotes where one level ends and the next one starts while the horizontal line denotes the result where the whole graph is considered a single component. $c=0.85,~\text{tol} = 10^{-9}$}
\label{fig:google_itr}
\end{figure} 
In Fig. \ref{fig:google_itr} we can see a couple of things, first the number of iterations over any component is less than the number of iterations that would be needed if we calculated PageRank of the graph as if it was a single huge component. The average number of iterations per edge was $148$, which can be compared to the number of iterations for the graph as a single component which was $168$, this gives an improvement of approximately $12\%$. This might look small looking at the figure, but remember that the largest components on each level are put first on their level and the size of components approximately follows a power law, hence large parts of the figure represent relatively few vertices. It should also be noted that a significant number of edges (approximately $26\%$) lies either between components or within CACs both of which are not counted for here since they don't use the iterative method and are instead used only once either to modify weights between levels or as part of the DFS when calculating PageRank of CACs. 

The second point of interest is that there is a clear difference between components at the last level compared to those of a higher level. Any SCC on the last level is by definition a stochastic matrix (before multiplication with $c$) since they have no edges to any vertex in any other component, this gives a lower bound on the number of iterations equal to $\log tol / \log c \approx 128 $ easily seen from the relation $c^\text{itr} \le \text{tol}$, where $\text{itr}$ is the number of iterations. However those component of higher level are by definition a sub stochastic matrix (before multiplication with c) since there is at least one edge to some other component. This is equivalent to some vertices having a lower $c$ value and the algorithm can therefor converge faster. 

The third observation is that a large component generally needs more iterations than a smaller component. This makes sense if we consider that as long as most vertices in the component do not contain edges to other components, as the component grows in size at least some part of the component will behave similar to those in the last level and we thus need a larger number of iterations.  For large components the estimated number of iterations is usually quite good, while for small components it usually gives a too high estimate (unless it is part of the last level). 

The running time in seconds for the Web graph for the three different methods for different values of error tolerance from $10^{-1}$ to $10^{-20}$ can be seen in Fig. \ref{fig:google1a}. 
\begin{figure}[!hbt]
\captionsetup{width=0.8\textwidth}
\centering
\begin{subfigure}[t]{ .4\textwidth}
\centering
\includegraphics[width = 1\textwidth]{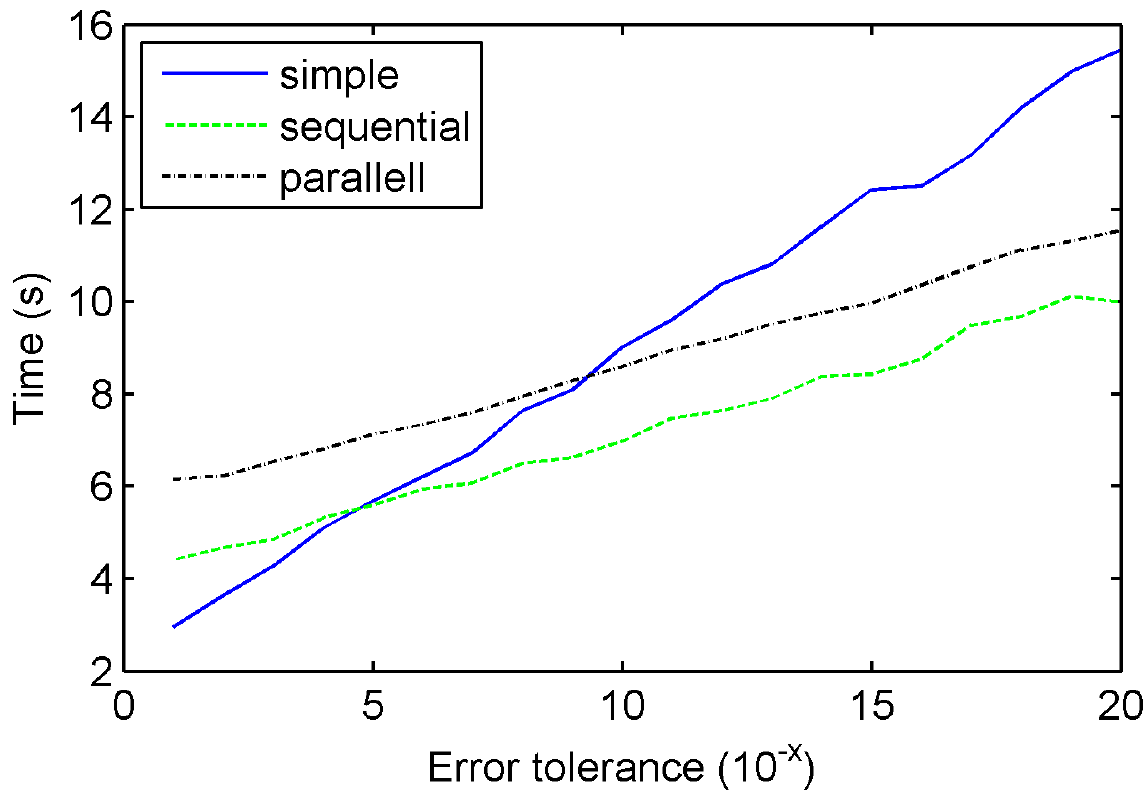}
\caption{} 
\label{fig:google1a}
\end{subfigure}%
~
\begin{subfigure}[t]{ .4\textwidth}
\centering
\includegraphics[width = 1\textwidth]{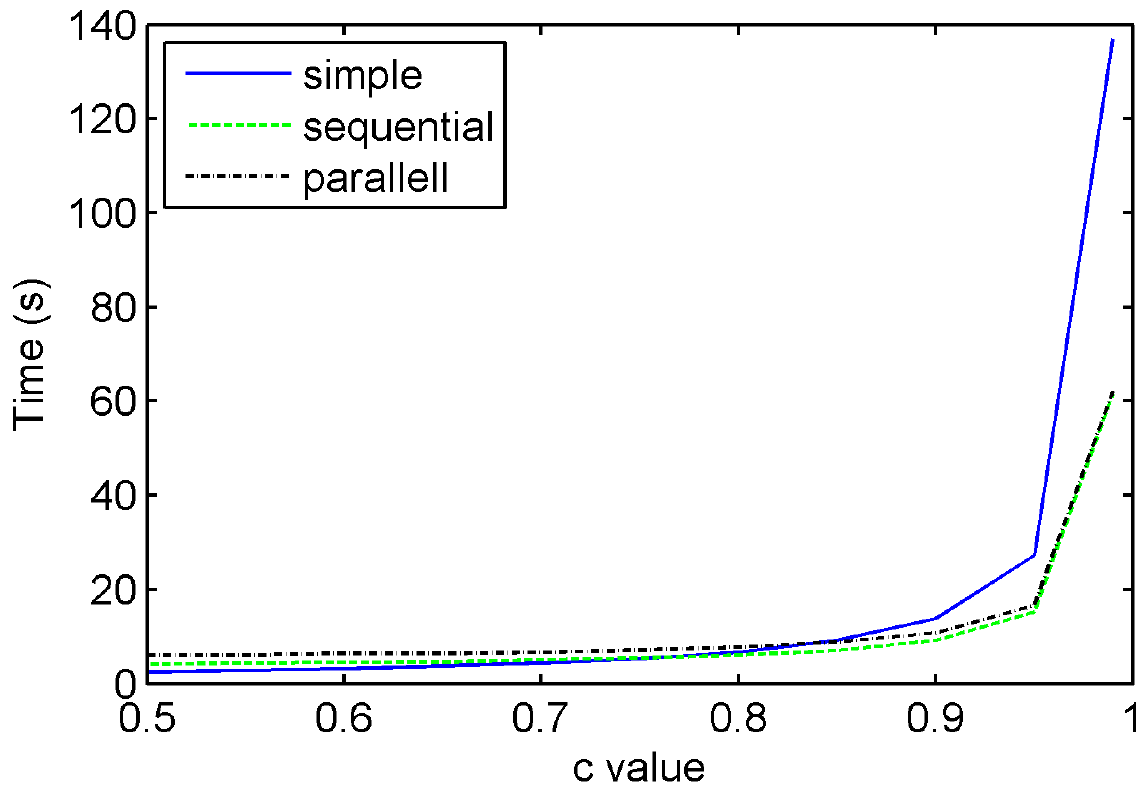}
\caption{} 
\label{fig:google1b}
\end{subfigure}
\caption{Running time needed to calculate PageRank for 3 different methods (a) depending on error tolerance using $c=0.85$ and (b) depending on $c$ between $0.5$ and $0.99$ using $\text{tol}= 10^{-10}$ on the Web graph. }
\label{fig:google1}
\end{figure}
%
%
From this it is clear that our method adds a significant amount of overhead, especially the parallel one. While all three methods need a longer time if the error tolerance is smaller, our method shows a significantly smaller increase compared to the basic method. While the break even here seems to lie at around $10^{-5}$ for the sequential algorithm and at around $10^{-9}$ for the parallel algorithm because of additional overhead. If the overhead in particular for the parallell algorithm could be further reduced this breakpoint could potentially be significantly earlier. Note that because of limits in machine precision we might not have the correct rank down to the last 20 decimals at the lowest tolerance, however since we sum over successively smaller parts we still get a good approximation of the actual computation time.

The running time when we let $c$ vary between $0.5$ and $0.99$ with a constant error tolerance ($10^{-10}$) can be seen in Fig. \ref{fig:google1b}.
%
%

From Fig. \ref{fig:google1b} there is a clear indication that our proposed method can be significantly faster for values of $c$ close to one. This is quite natural given that computation time of some components does not depend on $c$ at all (CACs and small components). Even for those components that does depend on error tolerance (large SCCs) there should be a significant amount of edges out of the component meaning some rows will have a sum lower than $c$ thus roughly simulating a lower $c$ value. 

In order to get some estimation of how the result changes with the size of the graph we also did the same experiments with ten copies of this graph. The results of this can be seen in Fig. \ref{fig:google2a}.
\begin{figure}[!hbt]
\captionsetup{width=0.8\textwidth}
\centering
\begin{subfigure}[t]{ .4\textwidth}
\centering
\includegraphics[width = 1\textwidth]{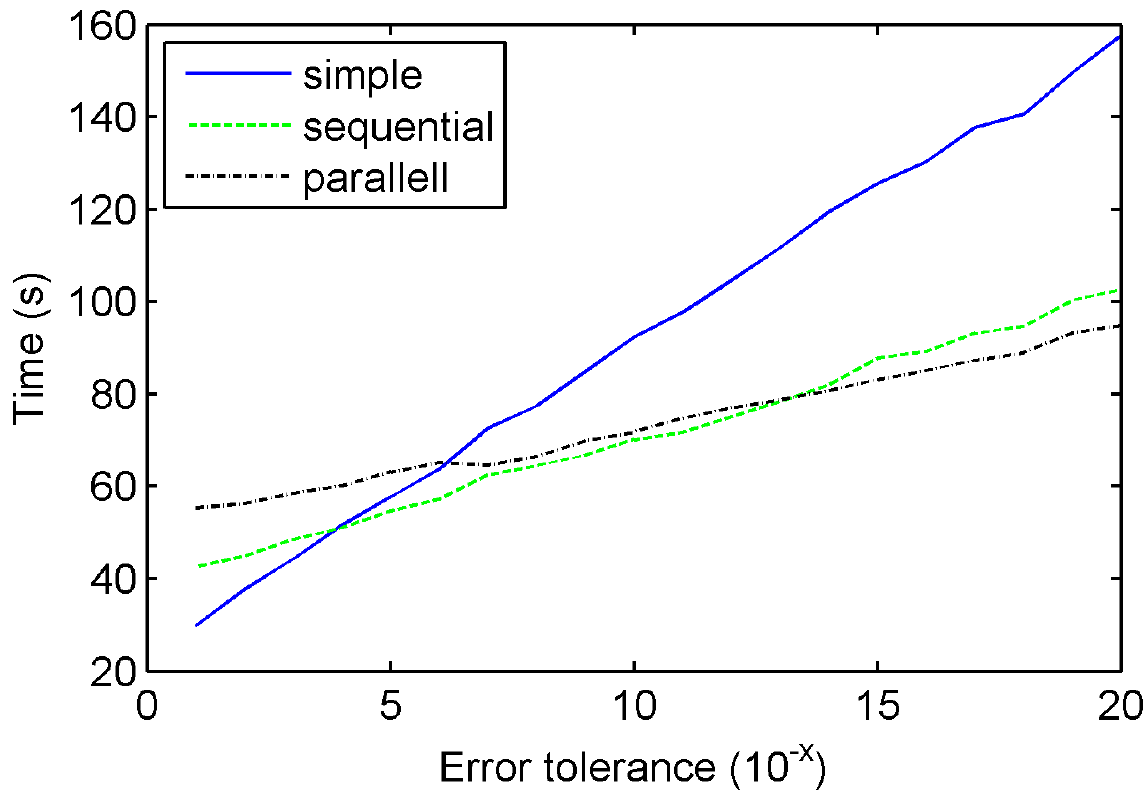}
\caption{} 
\label{fig:google2a}
\end{subfigure}%
~
\begin{subfigure}[t]{ .4\textwidth}
\centering
\includegraphics[width = 1\textwidth]{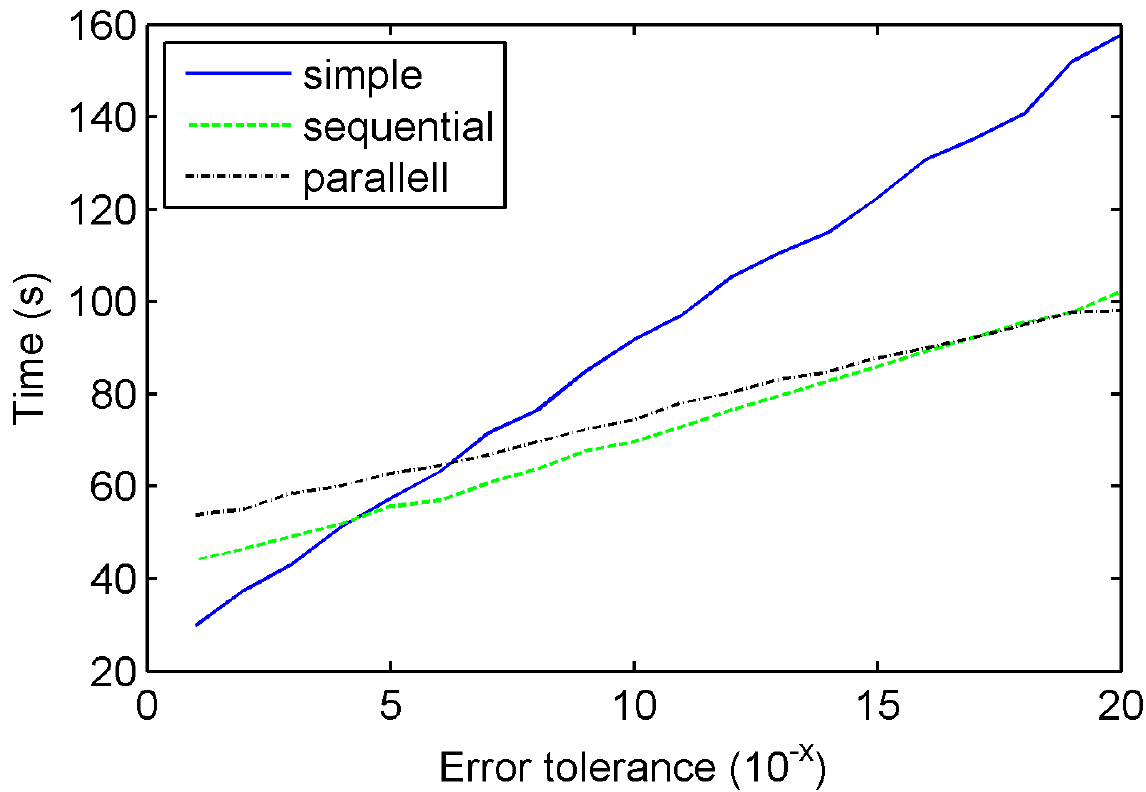}
\caption{} 
\label{fig:google2b}
\end{subfigure}
\caption{Running time needed to calculate PageRank for 3 different methods depending on error tolerance used on (a) ten copies of the Web graph and (b) ten copies of the Web graph with some extra edges added. }
\label{fig:google2}
\end{figure}
%
%
With a ten times larger graph the basic and sequential approach takes about ten times longer (sequential a little less, basic a little more), hence the relationship between the two are more or less the same but with a slightly earlier point at which they have the same performance. The parallel version however only took approximately eight times longer on the much larger graph presumably since we get many more components on the same level. This moves the breakpoint from $10^{-10}$ to $10^{-6}$.

While a ten-times copied graph might not represent a true network of that size it does have one important likeness, namely that the diameter of many real world networks (such as the world wide web) increases significantly slower than the size of the network itself. If the diameter is small it is clear that the number of levels need to be small as well, which means that as the size of the network increases we should have more and more vertices on the same level increasing the opportunity for the parallelization. 

By adding a couple of random edges we retain a single large component (as is the case of the original graph), the results after doing this can be seen in Fig. \ref{fig:google2b}.
%
%
We do not see any significant differences between this and the previous graph, the parallel method takes slightly longer while the other two stay more or less the same. From this we see that even if we have a very large component (composed of roughly half the vertices) we can still get significant gains using our method as long as the graph is sufficiently big. 

Last we also did the same experiment on a graph generated by the Barab\'asi-Albert model, the results for this graph can be seen in Fig. \ref{fig:BA}.
\begin{figure}[!hbt]
\captionsetup{width=0.8\textwidth}
\centering
\begin{subfigure}[t]{ .4\textwidth}
\centering
\includegraphics[width = 1\textwidth]{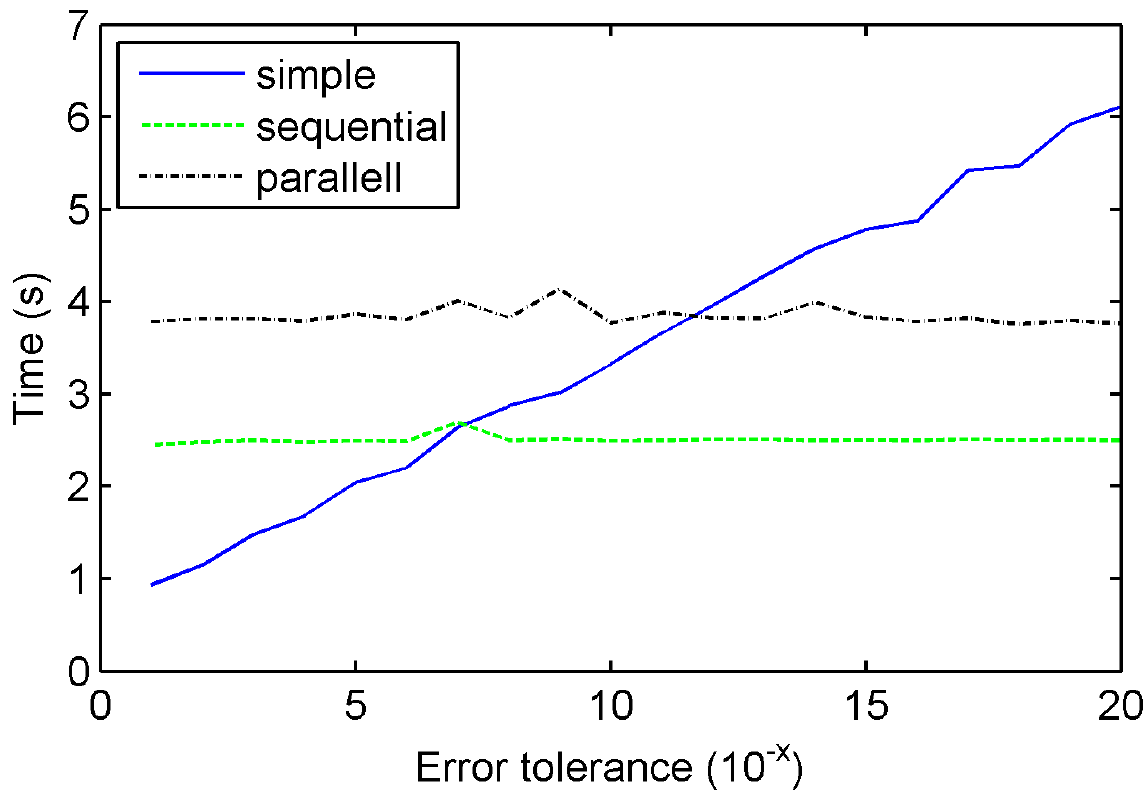}
\caption{} 
\label{fig:BA}
\end{subfigure}%
~
\begin{subfigure}[t]{ .4\textwidth}
\centering
\includegraphics[width = 1\textwidth]{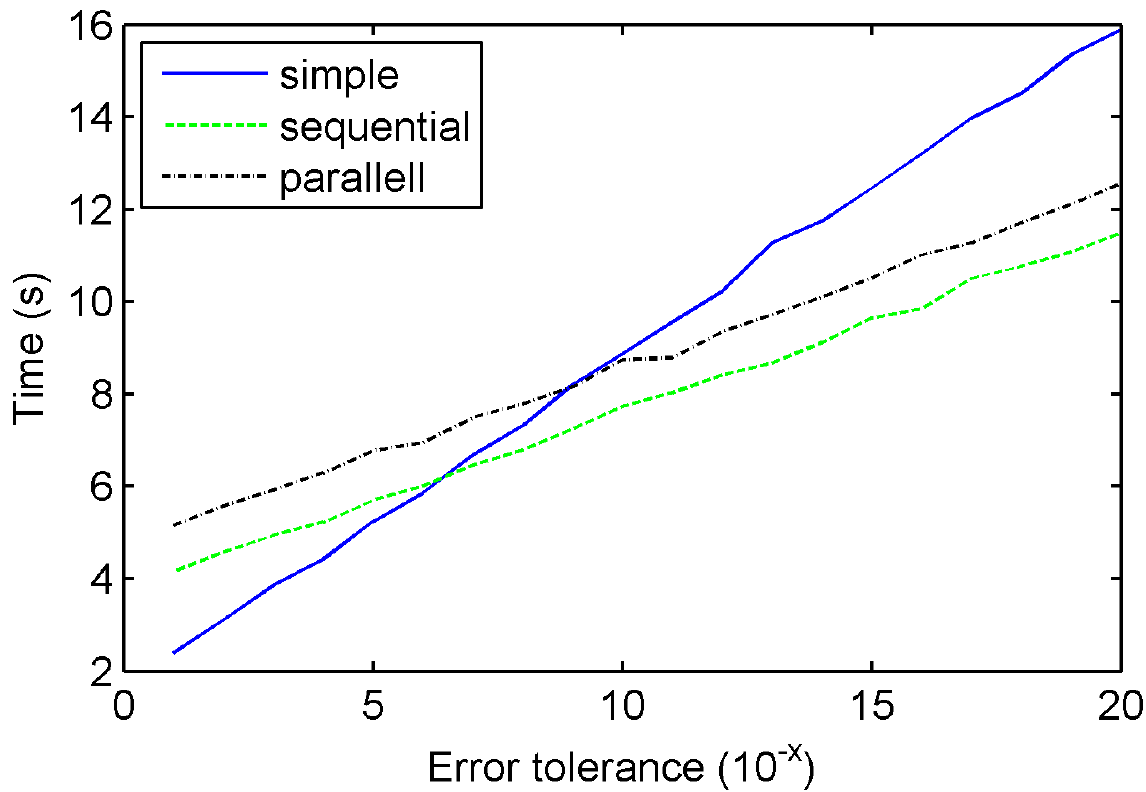}
\caption{} 
\label{fig:BAB}
\end{subfigure}
\caption{Running time needed to calculate PageRank for 3 different methods depending on error tolerance used on the graph generated by the Barab\'asi-Albert graph generation method with (a) roughly 1 edge/vertex after removal of edges and (b) roughly 5 edges/vertex after removal of edges.}
\label{fig:BA}
\end{figure}
%
Since this graph has a very small number of SCCs, both the sequential and parallel method barely depend on the error tolerance at all, while the basic approach seems to have similar behavior as with the previous graphs with a roughly linear increase in time needed as the error tolerance decreases. However since the PageRank calculations themselves are faster for this graph (for all methods), the overhead needed represent a larger part of the total time since this does not decrease in the same way. 
%
Overall the results of our approach are promising, we have gotten better results the bigger the graph is as well as the lower the error tolerance is as compared to the basic approach. Our implementation of the parallel method had significantly more overhead than the sequential method, this is likely to change if the algorithm was implemented in full (with parallelization in mind) in for example c++ where we have more control over how it can be implemented. 

\section{Conclusions}

We have seen that by dividing the graph into components we could get improved performance of the PageRank algorithm, however it does come with a cost of
 increased overhead (very much depending on implementation) as well as algorithm complexity. We could see that we needed to do significantly less number of iterations overall using our method compared to the ordinary method using a power series or power iterations. This came from the fact that some edges lie between components or inside a CAC (hence needing only one iteration) as well as the fact that a larger component is more likely to need a larger number of iterations than a small one. 

From our experiments we can see that our proposed algorithm is more effective compared to the basic approach as the size of the graph increases, at least as long as we can keep everything in memory. More experiments would be needed for any conclusions after that although given that components can be calculated by themselves in our method (hence a lower memory requirement in the PageRank step) we expect our method to compare even better when this is the case. We could also see better performance using our method compared to the basic method if the error tolerance is small while it is generally slower because of overhead if the tolerance for errors is large. 

The results for the parallel method compared to the serial method was not clear, however this is likely something that can be improved significantly using a better implementation and memory handling as well as having more of an advantage when memory overall is more of a concern for even larger graphs. 

By calculating PageRank for different kinds of components differently we could see a large improvement for certain types of graphs such as the one generated by the B-A model where the time needed to calculate PageRank was more or less constant depending on error tolerance using our method. 

\section{Future work}
One obvious and likely next step would be to implement and try the same component finding algorithm for something else such as sparse equation solving or calculating the inverse of sparse matrices, both of which have similar behavior in that components only affect other components downwards in levels. 

It would also be interesting to try an improved implementation of the method by implementing it all in c++ or another `fast' programming language in order to reduce some of the extra overhead incurred. Similarly it would also be interesting to see how the algorithm compares to other algorithms for extremely large graphs where memory become more of an issue. 

A third interesting direction would be to look at how this method interacts with other methods to calculate PageRank by using a method such as the one proposed by \cite{Kamvar200451} on the largest components. 

Another important problem is how to update PageRank after doing some changes to the graph, it would be very interesting to see how this partitioning of the graph in combination with non-normalized PageRank could potentially be used to recalculate PageRank faster than for example using the old PageRank as initial rank and doing power iterations from there. 

\bibliographystyle{plain}
\bibliography{ce-bib}

\end{document}